\documentclass[12pt]{article}
\usepackage[margin=1in]{geometry}
\usepackage{setspace}
\usepackage{amsfonts}       
\usepackage{amsthm}
\usepackage{dsfont}     
\usepackage{booktabs}
\usepackage{float}
\usepackage{graphicx}
\usepackage{hyperref}
\usepackage{mathtools}
\usepackage{microtype}
\usepackage{subfig}
\usepackage{thm-restate}
\usepackage{algorithm, algorithmic}
\usepackage{color}

\newtheorem{lemma}{Lemma}
\newtheorem{theorem}{Theorem}
\newtheorem{definition}{Definition}

\def\B{\mathcal{B}}
\def\S{\mathcal{S}}

\def\eps{\epsilon}
\def\lc{\left\lceil}   
\def\rc{\right\rceil}
\def\lf{\left\lfloor}   
\def\rf{\right\rfloor}

\def \O{\mathcal{O}}
\def \tO{\tilde{\mathcal{O}}}

\def \tDO{\widetilde{\mathcal{D}\mathcal{O}}}

\def\D{\mathcal{D}}
\def\H{\mathbb{H}}

\DeclareMathOperator{\arccosh}{arccosh}

\DeclareMathOperator*{\argmin}{arg\,min}
\DeclareMathOperator*{\argmax}{arg\,max}

\newcommand{\recenteringAlg}{\emph{Recentering-HyperbolicNN }}
\newcommand{\recenteringAlgNS}{\emph{Recentering-HyperbolicNN}}

\newcommand{\binarySearchAlg}{\emph{Binary-Search-HyperbolicNN }}
\newcommand{\binarySearchAlgNS}{\emph{Binary-Search-HyperbolicNN}}

\newcommand{\bandedAlg}{\emph{Spherical-Shell-HyperbolicNN }}
\newcommand{\bandedAlgNS}{\emph{Spherical-Shell-HyperbolicNN}}

\newcommand{\recenterBalls}{\emph{Euclidean-Center-of-Hyperbolic-Ball }}
\newcommand{\recenterBallsNS}{\emph{Euclidean-Center-of-Hyperbolic-Ball}}

\newcommand{\decisionProbeBand}{\emph{Check-Intersection }}
\newcommand{\decisionProbeBandNS}{\emph{Check-Intersection}}

\newcommand{\decisionProbeBandMaxRadius}{\emph{Choose-Band }}
\newcommand{\decisionProbeBandMaxRadiusNS}{\emph{Choose-Band}}

\newcommand{\bandingDatasetAlg}{\emph{Spherical-Shell-Partition }}
\newcommand{\bandingDatasetAlgNS}{\emph{Spherical-Shell-Partition}}

\newcommand{\bandedRandAlg}{\emph{Randomized-Spherical-Shell-HyperbolicNN }}
\newcommand{\bandedRandAlgNS}{\emph{Randomized-Spherical-Shell-HyperbolicNN}}

\title{Nearest Neighbor Search for Hyperbolic Embeddings}

\author{
	Xian Wu \\
	Stanford University \\ 
	\texttt{xwu20@stanford.edu}
	\and
	Moses Charikar \\
	Stanford University \\
	\texttt{moses@cs.stanford.edu}
	}

\date{}
\begin{document}

\maketitle

\begin{abstract}
Embedding into hyperbolic space is emerging as an effective representation technique for datasets that exhibit hierarchical structure. This development motivates the need for algorithms that are able to effectively extract knowledge and insights from datapoints embedded in negatively curved spaces. We focus on the problem of nearest neighbor search, a fundamental problem in data analysis. We present efficient algorithmic solutions that build upon established methods for nearest neighbor search in Euclidean space, allowing for easy adoption and integration with existing systems. We prove theoretical guarantees for our techniques and our experiments demonstrate the effectiveness of our approach on real datasets over competing algorithms. 
\end{abstract}

\section{Introduction}
We study the nearest neighbor problem for vector representations in hyperbolic space: given a dataset $\D$ of vectors and a query $q$, find the nearest neighbor of $q$ among the elements of $\D$ according to the hyperbolic distance metric. Nearest neighbor search is an important building block in many applications, including classification, recommendation systems, DNA sequencing, web search, and near duplicate detection. Yet for embeddings into negatively curved spaces, we still lack simple, practical, experimentally verified and theoretically justified solutions to tackle this question.

Hyperbolic embeddings have emerged as a useful way of representing data that exhibit hierarchical structure. \cite{nickel2017poincare} studies the representation and generalization performance of hyperbolic embeddings in comparison to Euclidean and translational embeddings and shows that hyperbolic embeddings outperforms with just a few dimensions. Later work focuses on techniques to produce even higher quality hyperbolic embeddings, including different training algorithms in different models of hyperbolic space \cite{nickel2018learning} and combinatorial embedding algorithms \cite{de2018representation}, and hybrid training models \cite{le2019inferring}. These developments motivate the need for algorithms that are able to effectively extract knowledge and insights from hyperbolic data representations, for example neural networks that can work with hyperbolic embeddings as feature vectors \cite{ganea2018hyperbolic}. We focus on the problem of nearest neighbor search. 

Despite the extensive literature on nearest neighbor search, most focus on the Euclidean setting and very few existing algorithms can be applied to hyperbolic embeddings. One relevant work for hyperbolic space is \cite{krauthgamer2006algorithms}, which proposes an approximate nearest neighbor search scheme that involves iteratively partitioning the space using special separator points.
They prove the existence of such points, but do not give an efficient algorithm to find them. Moreover, their solution requires precise a-priori knowledge of intrinsic parameters, such as the hyperbolicity of the dataset, that are computationally very difficult to compute exactly, \cite{borassi2015computing}. Their approximation guarantees are in terms of these parameters, so using upper bounds could lead to poor performance. 
There are also nearest neighbor graph methods \cite{malkov2018efficient} \cite{naidan2015permutation} \cite{fu2019fast} \cite{subramanya2019rand} that create a search graph for a dataset by linking elements are close together in a generic distance metric, and hyperbolic distance applies. The drawback is that they do not come with any theoretical guarantees and require a lot of hyperparameter tuning and high indexing costs. 

Our focus is on developing efficient nearest neighbor algorithms for hyperbolic space with provable guarantees that also work well in practice. We leverage solutions for provably efficient nearest neighbor search in Euclidean space and show how those algorithms can be used in a black box fashion to find nearest neighbors in hyperbolic space with minimal additional cost in query time and storage. Our solution is simple, intuitive, and easy to adopt by practitioners. We experiment on real datasets and show that our technique compares favorably against benchmark methods. Our theoretical analysis develops a rigorous understanding of our techniques and our ideas offer insights on key properties of negatively curved spaces that we hope will benefit future algorithmic work on hyperbolic space. 

\section{Related Work}
Our work adds to a fast-growing collection of exciting progress on hyperbolic representation learning, recently popularized by the work of \cite{nickel2017poincare} and \cite{nickel2018learning}. \cite{ganea2018hyperbolic} \cite{gu2018learning} \cite{law2019lorentzian} \cite{de2018representation} \cite{tifrea2018poincar} study techniques for learning more effective hyperbolic embeddings from hierarchical data, including both neural network and combinatorial based approaches. Works such as \cite{cho2019large} \cite{davidson2018hyperspherical} \cite{tran2020hyperml} develop techniques for performing downstream tasks such as classification and recommendation given pretrained embeddings. \cite{dhingra2018embedding} and \cite{tay2018hyperbolic} work in the NLP domain and train hyperbolic word embeddings and use them for downstream tasks such a Question Answering. \cite{chamberlain2017neural} embeds graphs into hyperbolic space. \cite{gulcehre2018hyperbolic} \cite{ganea2018hyperbolic} develop neural network architectures for transformers and recurrent neural networks that use hyperbolic geometry to learn from datasets with hierachical structure. 

Nearest neighbor methods in Euclidean space are well studied, see \cite{reza2014survey} for a general survey. There are many different techniques that come with provable guarantees, including Locality Sensitive Hashing \cite{wang2014hashing}, KD trees \cite{bentley1975multidimensional}, and many others, see \cite{reza2014survey} and references therein . On the empirical side, https://github.com/erikbern/ann-benchmarks compares performance of common nearest neighbor algorithms for benchmark datasets. However, these techniques and analyses are focused on Euclidean space, and do not apply immediately to hyperbolic space. To our knowledge, we are the first to present a theoretically justified and empirically validated solution for nearest neighbors in hyperbolic space.  

\section{Problem formulation and approach overview}
We are given a dataset $\D$ of $n$ points and a query $q$ in hyperbolic space and want to find the nearest neighbor or approximate nearest neighbor to $q$ from $p \in \D$. We call a point $p$ a $c$-approximate nearest neighbor for $c > 1$ if $d_H(p, q) \leq c \cdot d_H(p^*, q)$, where $p^*$ is the nearest neighbor to $q$ in the hyperbolic metric, and $d_H$ is the hyperbolic distance function. There are several models of hyperbolic space and we focus on the popular and intuitive Poincar\'e ball model in $r$ dimensions, which we denote $\H_r$. The different models are isometric, so one can apply our techniques to points embedded into other models by translating them to the Poincar\'e ball, see \cite{cannon1997hyperbolic} for details. 

\subsection{Preliminaries}

In $\H_r$, all points are inside the $r$-dimensional unit ball, and distance between points $x$ and $y$ is defined by 
\begin{equation}
\label{eq:poincare_distance}
d_H(x, y) = \arccosh \left(1 + \frac{2\|x-y\|^2}{(1-\|x\|^2) (1-\|y\|^2)}\right) ~,
\end{equation}
where $\|\cdot\|$ denotes Euclidean norm or Euclidean distance. 

We denote $\B_H(q, d)$ the hyperbolic ball around center $q$ with hyperbolic radius $d$. We denote $\B_E(q, d)$ as the Euclidean ball around center $q$ with Euclidean radius $d$. One useful fact is that for every $q, d$, $\B_H(q, d) = \B_E(q', d')$ for some $q', d'$ that can be solved via simple calculations (ie, hyperbolic balls in Poincar\'e space are Euclidean balls with different centers and radii) \cite{cannon1997hyperbolic}. 

\subsection{Overall approach}
Our overall approach is to leverage existing Euclidean nearest neighbor methods to find near exact hyperbolic nearest neighbors. {\emph{Our first class of algorithms use the key fact that hyperbolic balls in $\H_r$ are Euclidean balls with different centers of gravity.}} For query $q$, if we had $p \in \D$ such that $p \in \B_H(q, d_H(p, q)) = \B_E(q', d')$, then we can find a better neighbor by doing Euclidean nearest neighbor search on $q'$. 

{\emph{Our second main class of algorithms uses the insight that when $p \in \D$ have similar Euclidean norms, the denominator term $(1-\|p\|^2) (1-\|q\|^2)$ in Eq. \ref{eq:poincare_distance} is similar for different $p$, so the problem reduces to minimizing $2\|p-q\|^2$, which is a Euclidean nearest neighbor problem.}} We first partition our dataset so that elements in one partition have similar Euclidean norms, perform Euclidean nearest neighbor search in these partitions separately, and then aggregate results. For massive datasets, this idea also provides a way to shard the database that maintains efficient search and indexing. 

We abstract our use of Euclidean nearest neighbor algorithms into black box oracles; our algorithms are compatible with any implementation of Euclidean nearest neighbor search, however performance varies depending on the underlying algorithm. We use the following classes of oracles:

\begin{definition}[Exact Euclidean Nearest Neighbor Oracle $\O$]
The exact Euclidean nearest neighbor oracle, $\O$ takes as input a query $q$ and a dataset $\D$ and returns $O(q, \D)$, which is an element $d$ in $\D$ that minimizes Euclidean distance to $q$ in query time $\mathcal{T}$ and space $\S$. 
\end{definition}
\begin{definition}[$(1+\eps)$-approximate Euclidean Nearest Neighbor Oracle $\tO$]
For $\eps > 0$, a $(1+\eps)$-approximate Euclidean nearest neighbor oracle, $\tO$ takes as input a query $q$ and a dataset $\D$, and returns $\tO(q, \D)$, which is some $d \in \D$ such that $\|d - q\| \leq (1+\eps) \|q-n_E\|$ in query time $\mathcal{T}$ and space $\S$, where $n_E$ is the Euclidean nearest neighbor to $q$ in $\D$. 
\end{definition}

We do not include failure probability into our definition of $\tO$ even though many of them give high probability guarantees, because this can be resolved using independent trials. Examples of common oracles and their performance are in \cite{wang2014hashing} and references therein.

To summarize, our main contributions are: 
\begin{itemize}
\item \recenteringAlgNS, an exact hyperbolic nearest neighbor algorithm that uses an exact Euclidean nearest neighbor oracle. 
\item \binarySearchAlgNS, a $c$-approximate hyperbolic nearest neighbor algorithm that uses an exact Euclidean nearest neighbor oracle. 
\item \bandedAlgNS, a $c$-approximate hyperbolic nearest neighbor algorithm that uses a $(1+\eps)$-approximate Euclidean nearest neighbor oracle. 
\end{itemize}

\section{Recentering algorithms using exact Euclidean oracles}

In each iteration of \recenteringAlgNS, Algorithm \ref{alg:nn_recentering_alg}, we take the current best hyperbolic nearest neighbor $n_H$ (initially set to be the Euclidean nearest neighbor of $q$) and attempt to find a closer point in hyperbolic distance. We exploit the fact that the hyperbolic ball around $q$ that has $n_H$ on its boundary is a Euclidean ball around a different point $q_{new}$ \cite{cannon1997hyperbolic}. Performing Euclidean nearest neighbor search around $q_{new}$ either finds a point strictly inside this ball (which is closer to q than $n_H$ in hyperbolic distance), or establishes that $n_H$ indeed is the hyperbolic nearest neighbor of $q$. 

\recenteringAlg uses \recenterBallsNS, an elementary subroutine that performs the recentering. Details can be found in \cite{cannon1997hyperbolic} and in the appendix.  

\begin{algorithm}[ht]
\caption{\recenteringAlg}
\label{alg:nn_recentering_alg}
\begin{algorithmic}[1]
\REQUIRE{query $q$, dataset $\D$, exact Euclidean nearest neighbor oracle $\O$}
\STATE $n_H \leftarrow q$
\STATE $n_E \leftarrow \O(q, \D)$. 
\WHILE{$d_H(n_E, q) \neq d_H(n_H, q)$}
\STATE $n_H \leftarrow n_E$
\STATE $q_{new} = \recenterBalls(q, d_H(q,n_H))$
\STATE $n_E \leftarrow \O(q_{new}, \D)$
\ENDWHILE
\STATE \textbf{Return} $n_H$
\end{algorithmic}
\end{algorithm}

\begin{theorem}
\label{thm:nn_recentering_alg}
Suppose that the Euclidean nearest neighbor to $q$, is the $k$-th nearest hyperbolic neighbor to $q$. Then Algorithm \ref{alg:nn_recentering_alg} returns the hyperbolic nearest neighbor $n_H$ after at most $k+1$ invocations of the exact Euclidean nearest neighbor oracle $\O$. The runtime of this algorithm is at most $(k+1) \mathcal{T}$, where $\mathcal{T}$ is the runtime for one invocation of $\O$. The storage of this algorithm is $\S$, where $\S$ is the storage requirement of $\O$. 
\end{theorem}

\begin{proof}
If there is an exact match, we would invoke $\O$ once. If there is no exact match, the first invocation of $\O$ returns the Euclidean nearest neighbor to the query, $n_E$. We can draw a hyperbolic ball around $q$ with radius $d_H(q, n_E)$. Clearly, any point that is closer in hyperbolic distance to $q$ must lie inside this ball. So we will find these points by calling $\O$ on the Euclidean center of this ball, $q_{new}$, which guarantees an improvement over $n_E$. We recurse on this logic. If at round $r$, we do not get an improvement, then we terminate, as there cannot be a point that is a nearer neighbor. 

Since each round results in a strict improvement or a termination, if the Euclidean nearest neighbor of $q$ is the $k$-th nearest hyperbolic neighbor to $q$, then \recenteringAlg terminates in at most $k+1$ invocations of $\O$. The runtime guarantee follows trivially. 
\end{proof}

\recenteringAlg generalizes to provably return $K$ nearest neighbors using an oracle that finds $K$ Euclidean nearest neighbors when the recentering and termination criterion use the $K$-th nearest neighbor found so far. 

Theorem \ref{thm:nn_recentering_alg} provides a worst case guarantee in terms $k$, the ranking of the Euclidean nearest neighbor with respect to the hyperbolic metric. Our algorithm doesn't need to know $k$; moreover, in the best case, the datapoints could be such that the Euclidean nearest neighbor of $q$ is the $k$-th hyperbolic nearest neighbor to $q$ for arbitrarily high $k$ but \recenteringAlg returns the hyperbolic nearest neighbor in 3 invocations to $\O$. 

However, in the worst case, \recenteringAlg returns the hyperbolic nearest neighbor in exactly $k+1$ invocations of $\O$ for arbitrary $k$. We give the construction below. 

\begin{lemma}
Let $q$ be our query in 1 dimension, and $\|q\|$ is close to 1. Suppose for arbitrary $k \in \mathbb{N}$, we have data points $q+z, q-z, p_1, \ldots p_{k-2}$, where $p_i = \frac{2^i -1}{2^i}$, and $z$ is very small and satisfies $q - z \geq \frac{2^{k} -1}{2^{k}}$ and $q + z < 1$, and $d_H(q, 0) = d_H(q, q+z)$. Then \recenteringAlg returns $q-z$, hyperbolic nearest neighbor in exactly $k+1$ invocations to $\O$.
\end{lemma}
\begin{proof}
\recenteringAlg first returns $n_E = q + z$, the $k$-th hyperbolic nearest neighbor to $q$. $n_E$ is close to the edge of the disk whereas all the other points are closer to the origin, so $d_H(q, n_E)$ is high even though the Euclidean distance is small. The new center from the first recentering is near the point $\frac{1}{2}$, so the next call to $\O$ returns $p_1 = \frac{1}{2}$. Subsequent calls to $\O$ will return $\frac{3}{4} = p_2$, and then $p_3, \ldots p_{k-2}$ until we finally find $q-z$. 
\end{proof}

\subsection{$k$-Independent approximate hyperbolic nearest neighbor algorithm}

\binarySearchAlg is an approximate hyperbolic nearest neighbor algorithm that aims to approximate the smallest possible radius $r$ around the query such that $B_H(q, r)$ is non-empty, which essentially isolates the nearest neighbor. It performs binary search on $r$, starting from the upper bound $r = d_H(q, n_E)$, and continues until it finds a small enough non-empty radius that satisfies the desired approximation guarantee.  

Using the same recentering idea, we can use $\O$ to determine whether $B_H(q, r)$ is non-empty for any $r$. The nearest neighbor that $\O$ outputs is the certificate that indicates whether to recurse on the left or right side of the binary search. The algorithm maintains upper and lower bounds $R_i$ and $L_i$ on $r$ in each round $i$, ensuring that $\frac{R_{i+1}}{L_{i+1}}\leq\sqrt{R_i/L_i}$. 

\begin{algorithm}[ht]
\caption{\binarySearchAlg}
\label{alg:binary_search}
\begin{algorithmic}[1]
\REQUIRE{query $q$, exact Euclidean nearest neighbor oracle $\O$, approximation guarantee $c > 1$}
\STATE $n_E \leftarrow \O(q, \D)$
\STATE $n_H \leftarrow n_E$
\IF {$n_E = q$}
\STATE \textbf{Return} $n_H$
\ENDIF
\STATE $L = d_H\left(q, \left(1-\frac{\|d_E -q\|}{\|q\|}\right)\cdot q\right)$
\STATE $R = d_H(n_E, q)$
\WHILE{$R > cL$}
\STATE $q_{new} = \recenterBalls(q, \sqrt{RL})$
\STATE $n_E \leftarrow \O(q_{new}, \D)$
\IF {$d_H(n_E, q) > \sqrt{RL}$}
\STATE $L \leftarrow \sqrt{RL}$
\ELSE
\STATE $n_H \leftarrow n_E$
\STATE $R \leftarrow d_H(n_H, q)$
\ENDIF
\ENDWHILE
\STATE \textbf{Return} $n_H$
\end{algorithmic}
\end{algorithm}

\begin{theorem}
\label{thm:binary_search}
Given query $q$, and approximation constant $c > 1$, and letting $R_{initial} = d_H(q, n_E), L_{initial}$ be initial non-zero upper and lower bounds on the distance of the hyperbolic nearest neighbor to $q$, \binarySearchAlg returns a $c$-approximate hyperbolic nearest neighbor in at most $\log_2 \left( \frac{\log \left( \frac{R_{initial}}{L_{initial}}\right)}{\log(c)}\right)$ rounds. The total runtime is $\mathcal{T} \cdot \log_2 \left( \frac{\log \left( \frac{R_{initial}}{L_{initial}}\right)}{\log(c)}\right)$, where $\mathcal{T}$ is the runtime for one invocation of $\O$. The storage of this algorithm is $\S$, where $\S$ is the storage requirement of $\O$. 
\end{theorem}

\begin{proof}
$n^*$, the hyperbolic nearest neighbor to $q$, is always within hyperbolic distance $L$ and $R$ to $q$ in every iteration. This is true at the beginning of the algorithm: $d_H(q, n^*) \leq d_H(q, n_E) = R$. $\O$ produces a Euclidean nearest neighbor $n_E$, which means that the interior of the Euclidean ball around $q$ with radius $n_E$ is empty. Therefore $d_H(q, n^*)$ must be at least as far away from $q$ as the closest point on this ball to $q$ in hyperbolic distance. The closest. point can be expressed as $t\cdot q$ for $0< t < 1$, and satisfies $\|q - t\cdot q\| = \|d_E, q\|$. Therefore, this point is $\left(1-\frac{\|d_E -q\|}{\|q\|}\right)\cdot q$, and so $L$ as initialized in the algorithm is a valid lower bound. In the first iteration of the algorithm, we search within the hyperbolic ball around $q$ with hyperbolic radius $\sqrt{RL}$ by finding the Euclidean center to this ball and searching for the Euclidean nearest neighbor. If we find $n_E$ such that $d_H(q, n_E) \leq \sqrt{RL}$, this means that this ball is nonempty and so $n^*$ must be within hyperbolic distance $L$ and $\sqrt{RL}$. Furthermore, we have a point $n_E$ such that $L \leq d_H(q, n_E) \leq \sqrt{RL}$. Otherwise if this ball is empty then the nearest neighbor must have hyperbolic distance at least $\sqrt{RL}$ and so we update the lower threshold, $L$. Therefore, at any point in the algorithm, $L$ and $R$ represent valid upper and lower bounds for $d_H(q, n^*)$. Note also that the current $n_H$ is always a point such that $L \leq d_H(q, n_H) \leq R$. At each iteration, the square root the ratio $\frac{R}{L}$ from the previous round until we hit the termination condition that $\frac{R}{L} \leq c$, so that $n_H$ is a $c$-approximate nearest neighbor. 

Let $R_i$ and $L_i$ be the upper and lower thresholds at round $i$. Then in the next round, $\frac{R_{i+1}}{L_{i+1}} \leq \sqrt{\frac{R_i}{L_i}}$. Suppose the algorithm starts off with $R_{initial}$ and $L_{initial}$. Then \binarySearchAlg terminates in $\delta$ rounds, where $\left( \frac{R_{initial}}{L_{initial}} \right)^{\frac{1}{2^\delta}}\leq c$. Solving for $\delta$ yields $\delta \geq \log_2 \left( \frac{\log \left( \frac{R_{initial}}{L_{initial}}\right)}{\log(c)}\right)$. 
\end{proof}
In the worst case, we establish in Lemma \ref{lem:R_L_high} that $\frac{R_{initial}}{L_{initial}}$ can be arbitrarily high. The construction is simple -- we choose (Euclidean) co-linear $n_E, n_H, q$ where $\|n_E - q\| = \|n_H - q\|$, and show that the ratio can be arbitrary bad as the points approach the edge of the disk. 
Even though the ratio can become arbitrarily high as points approach the edge of the disk, for finite datasets, we prove the upper bound $\frac{R_{initial}}{L_{initial}} \leq O(\ln(\frac{1}{1-\|q\|^2}) + \ln(\frac{1}{1-\|n_E\|^2}))$. Practitioners can understand how long \binarySearchAlg might take in the worst case with some prior knowledge on the largest $\|x\|^2$ for $x \in \D$ in their dataset. We formalize this in Lemma \ref{lem:precision}.

\begin{lemma}
\label{lem:R_L_high}
Fix large $s$, and let $\gamma, \delta$ be such that $0 < \gamma < \delta < 1$, $\delta^{s+1} < \gamma < \delta^s$, and $\frac{\delta - 2\delta^s}{\delta + \delta^s} \geq \frac{1}{2}$. Further let the query $q = (0, 1- \frac{\gamma + \delta}{2})$, $n_E = (0, 1-\gamma)$, and $n_H = (0, 1-\delta)$. Then $\frac{R_{initial}}{L_{initial}} = \Omega (s)$. 
\end{lemma}

\begin{proof}
We start with the following 3 points: $n_E = (0, 1-\gamma)$, $q = \left(0, 1- \left(\frac{\gamma + \delta}{2}\right)\right)$, $n_H = (0, 1-\delta)$, where $0<\gamma <\delta < 1$. $q$ is exactly the midpoint between $n_E$ and $n_H$ in the Euclidean metric.

Now fix some very large constant $s$ where $s > 1$. Suppose that $\gamma$ is small enough that $\delta^{s+1} < \gamma < \delta^{s}$. Further suppose $s$ is large enough that $\frac{\left( \delta - 2\delta^s\right)}{\delta + \delta^s} \geq \frac{1}{2}$. We will show that when $\gamma$ and $\delta$ satisfy this regime, $d_H(n_E, q) / d_H(n_H, q) = \Omega(s)$, so to make this ratio very high, one can use a very large $s$. 

From \eqref{eq:poincare_distance}, we have $d_H(n_E, q) \geq \arccosh \left(1 + \frac{2 \left( \frac{\delta - \delta^s}{2}\right)^2}{(2\delta^s)(\delta + \delta^s)} \right)$. Note that:
\begin{align*}
\frac{\left(\frac{\delta - \delta^s}{2} \right)^2}{\delta + \delta^s}&= \frac{ \left( \frac{\delta}{2}\right)^2 - 2 \left( \frac{\delta}{2}\right) \left( \frac{\delta^s}{2}\right) + \left( \frac{\delta^s}{2}\right)^2}{\delta + \delta^s}\\
&\geq \frac{\frac{\delta^2}{4} - \frac{\delta^{s+1}}{2}}{\delta + \delta^s} = \frac{\frac{\delta}{4} \left( \delta - 2\delta^s\right)}{\delta + \delta^s} \geq \frac{\delta}{8}
\end{align*}
Therefore,
$$d_H(n_E, q) \geq \arccosh \left(1 + \frac{\delta}{8\delta^s} \right) = \arccosh \left(1 + \frac{1}{8\delta^{s-1}} \right)$$
$$d_H(n_H, q) \leq \arccosh \left(1 + \frac{2 \left( \frac{\delta - \delta^{s+1}}{2}\right)^2}{(\delta^2)\left(\frac{\delta + \delta^{s+1}}{2}\right)^2} \right) \leq \arccosh \left(1+ \frac{2}{\delta^2} \right)$$
Using the identity $\arccosh(x) = \ln (x + \sqrt{x^2-1})$, we have: 
$$d_H(n_E, q) \geq \ln \left(1 + \frac{1}{8\delta^{s-1}} \right) \geq \ln \left(\frac{1}{8\delta^{s-1}} \right)$$
$$d_H(n_H, q) \leq \ln \left( 2+ \frac{4}{\delta^2} \right) \leq \ln \left(\frac{6}{\delta^2} \right)$$

To conclude, we have: 
$$ \frac{d_H(n_E, q)}{d_H(n_H, q)} \geq \frac{\ln \left( \frac{1}{8\delta^{s-1}}\right)}{\ln \left(\frac{6}{\delta^2} \right)} = \frac{\ln \left( \frac{1}{\delta^{s-1}}\right) + \ln \left(\frac{1}{8}\right)}{\ln \left(\frac{1}{\delta^2} \right) + \ln(6)} = \frac{s-1}{2} + o(1)$$
Therefore, $ \frac{d_H(n_E, q)}{d_H(n_H, q)} =\frac{R_{initial}}{L_{initial}} = \Omega(s)$. 
\end{proof}

\begin{lemma}
\label{lem:precision}
\binarySearchAlg returns a c-approximate hyperbolic nearest neighbor in at most $\log_2 ((\log_2 b)/(\log c)) + O(1)$ rounds, where $b$ is the number of bits used to represent an arbitrary $x \in \D$.
\end{lemma}

\begin{proof}
We show $\frac{R_{initial}}{L_{initial}} \leq O\left(\ln(\frac{1}{1-\|q\|^2}) + \ln(\frac{1}{1-\|n_E\|^2}) \right)$. 
Let $q$ be the query, $n_E$ be the Euclidean nearest neighbor to $q$, and $n_H$ be the point such that $\|q - n_H\| = \|q - n_E\|$ and $d_H(q, n_H)$ is minimized. This maximizes $\frac{R_{initial}}{L_{initial}}$. Let $\eps = \|q - n_H\| = \|q - n_E\|$, let $\delta = 1-\|n_E\|^2$. 

First we assume the case that $\|n_E\| \geq 2\eps$. We have: $\|q\| \geq \|n_E\| - \eps$, so that $1 - \|q\|^2 \leq 1-\|n_e\|^2 + 2\eps\|n_E\| - \eps^2 \leq 1 - \|n_E\|^2 + 2\eps = \delta + 2\eps$. We also have that $\|n_H\| \geq \|n_E\| - 2\eps$, therefore $1-\|n_H\|^2 \leq 1 - \|n_E\|^2 + 4\eps \|n_E\| - 4\eps^2 \leq \delta + 4\eps$.  We can write $d_H(q, n_H) \geq \arccosh \left( 1 + \frac{2\eps^2}{(\delta + 2\eps)(\delta + 4\eps)}\right)$. Therefore, if $\delta < \eps$, then $d_H(q, n_H) = \Omega (1)$, and $d_H(q, n_E) = O\left(\ln(\frac{1}{1-\|q\|^2}) + \ln(\frac{1}{1-\|n_E\|^2}) \right)$. If $\delta \geq \eps$, then $\frac{d_H(q, n_E)}{d_H(q, n_H)} = \frac{\arccosh(1 + f_1)}{\arccosh(1+f_2)}$, where $\frac{f_1}{f_2} = \frac{\delta + 4\eps}{\delta} \leq 5$, so we conclude that $\frac{d_H(q, n_E)}{d_H(q, n_H)} = O(1)$ in this case. 

We consider the case that $\|n_E\| < 2\eps$. Suppose $1-\|n_E\|^2 = \delta > \frac{1}{2}$. Then $\frac{d_H(q, n_E)}{d_H(q, n_H)} = \frac{\arccosh(1 + f_1)}{\arccosh(1+f_2)}$, where $\frac{f_1}{f_2} = \frac{1-\|n_H\|^2}{1-\|n_E\|^2} \leq 2$, so $\frac{d_H(q, n_E)}{d_H(q, n_H)} = O(1)$ in this case. When $1-\|n_E\|^2 = \delta \leq \frac{1}{2}$, it follows that $\frac{1}{\sqrt{2}} \leq \|n_E\| < 2\eps$, so $\eps > \frac{1}{2\sqrt{2}}$. Therefore, $d_H(q, n_H) \geq \arccosh(1 + \frac{1}{4}) = \Omega (1)$. Therefore \\ $\frac{d_H(q, n_E)}{d_H(q, n_H)} = O\left(\ln(\frac{1}{1-\|q\|^2}) + \ln(\frac{1}{1-\|n_E\|^2}) \right)$ in this case. 
\end{proof}

\subsubsection{Integration with approximate Euclidean nearest neighbor oracles}
Since approximate nearest neighbor algorithms are heavily used, we consider \recenteringAlg and \binarySearchAlg when powered by approximate Euclidean nearest neighbor oracles $\tO$. We show, somewhat surprisingly, that replacing the exact Euclidean nearest neighbor oracle by an approximate oracle can cause both algorithms to return points with arbitrarily bad approximation ratios. The next section shows how approximate oracles can be used to derive approximate hyperbolic nearest neighbor algorithms.

\begin{lemma}
For any $\eps > 0$, \recenteringAlg using a $(1+\eps)$-approximate Euclidean nearest neighbor oracle $\tO$ can return an approximate hyperbolic nearest neighbor with an arbitrarily bad approximation ratio. 
\end{lemma}

\begin{proof}
Suppose $q = (0, y)$, and $n_E = (0, y+r)$ and $n_H = (0, y-r)$ for $r > 0$. 

Then we have: 
\[
d_H(q, n_E) = \arccosh \left( 1 + \frac{2r^2}{(1-y^2)(1-(y+r)^2)} \right)
\]

The bottom of the hyperbolic circle with radius $d_H(q, n_E)$ is a point $B = (0, y-b)$ for $0<b$ that satisfies: 
\[
d_H(q, B) = d_H(q, n_E) = \arccosh \left( 1 + \frac{2b^2}{(1-y^2)(1-(y-b)^2)} \right)
\]

The Euclidean center, denoted $n_c$ is $\frac{y+r+y-b}{2} = \frac{2y+r-b}{2} = y + \frac{r-b}{2}$. 

In order for \recenteringAlg to fail with a $(1+\eps)$-Euclidean oracle, we want $y-r > y + \frac{r-b}{2} + \frac{d_E(n_c, n_E)}{1+\eps}$, where $d_E(n_c, n_E) = \frac{r+b}{2}$. This means that we want $b > \frac{r(4+3\eps)}{\eps}$. 

We want $b$ such that 
\[
\frac{r^2}{1-(y+r)^2} = \frac{b^2}{1-(y-b)^2}
\]
This implies that $b = \frac{r-ry^2}{1-y^2-2ry}$. Combined with the condition that $b > \frac{r(4+3\eps)}{\eps}$, we want: 
\[
\frac{1-y^2}{1-y^2-2ry} > \frac{4}{\eps} + 3
\]

Now we substitute in $y = 1 - \frac{\delta + \gamma}{2}$ and $r = \frac{\delta - \gamma}{2}$, and we maintain the condition that $\delta^{s+1} < \gamma < \delta^{s}$. 
This implies: 
\begin{align*}
\frac{1-y^2}{1-y^2-2ry} &= \frac{\delta + \gamma - \left( \frac{\gamma + \delta}{2}\right)^2}{\delta + \gamma - \left( \frac{\gamma + \delta}{2}\right)^2 - 2 \left( \frac{\delta - \gamma}{2}\right)\left(1 - \frac{\delta + \gamma}{2} \right)} \\
& \geq \frac{\delta + \delta^{s+1} - \left( \frac{\delta^{s} + \delta}{2}\right)^2}{\delta + \gamma - \left( \frac{\gamma + \delta}{2}\right)^2 - 2 \left( \frac{\delta - \gamma}{2}\right) + 2\left(\frac{\delta - \gamma}{2} \right)\left(\frac{\delta + \gamma}{2} \right)} \\
& =  \frac{\delta + \delta^{s+1} - \left( \frac{\delta^{s} + \delta}{2}\right)^2}{2\gamma - \left( \frac{\gamma + \delta}{2}\right)^2 + 2\left(\frac{\delta - \gamma}{2} \right)\left(\frac{\delta + \gamma}{2} \right)} \\
& \geq  \frac{\delta + \delta^{s+1} - \left( \frac{\delta^{s} + \delta}{2}\right)^2}{2\gamma + \left( \frac{\gamma + \delta}{2}\right)^2}
\end{align*}

Note that since $\left(\frac{\delta + \gamma}{2}\right)^2 = \left( \frac{\delta}{2}\right)^2 + \frac{\delta \cdot \gamma}{2} + \left( \frac{\gamma}{2}\right)^2 \leq \left( \frac{\delta}{2}\right)^2 + \frac{\delta^{s+1}}{2} + \left( \frac{\delta^s}{2}\right)^2$
We therefore have, 
\[
\frac{1-y^2}{1-y^2-2ry} \geq \frac{\delta + \delta^{s+1} - \left( \frac{\delta^{s} + \delta}{2}\right)^2}{2\delta^s + \left( \frac{\delta}{2}\right)^2 + \frac{\delta^{s+1}}{2} + \left( \frac{\delta^s}{2}\right)^2} = \frac{\delta + o(\delta)}{\frac{\delta^2}{4} + o(\delta^2)} = \Theta \left(\frac{1}{\delta}\right)
\]

Suppose that $\frac{1-y^2}{1-y^2-2ry} \geq \frac{k_1}{\delta}$ for some $k_1>0$. Then we need $\delta$ such that 
\[
\frac{k_1}{\delta} > \frac{4}{\eps} + 3
\]
This implies that $\delta < \frac{k_1 \cdot \eps}{4 + 3\eps}$. So for sufficiently small $\delta$, \recenteringAlg will fail to find $n_H$ during the recentering phase. Moreover, for sufficiently small $\eps$, given $\delta$, the ratio $d_H(q, n_E) / d_H(q, n_H)$ can be arbitrarily high. Therefore we conclude that \recenteringAlg with a $(1+\eps)$ approximate Euclidean oracle can return an answer with arbitrarily high approximation ratio.
\end{proof}

\begin{lemma}
For any $\eps > 0$, \binarySearchAlg using a $(1+\eps)$-approximate Euclidean nearest neighbor oracle $\tO$ can return an approximate hyperbolic nearest neighbor with an arbitrarily bad approximation ratio. 
\end{lemma}

The proof is similar and is in the appendix. 

\section{Approximate Near Neighbors}
The previous section shows that \recenteringAlg and \binarySearchAlg cannot guarantee a close hyperbolic nearest neighbor when using $\tO$, an approximate Euclidean nearest neighbor oracle. We now develop \bandedAlgNS, which uses $\tO$ to return neighbors with provable guarantees on the hyperbolic approximation ratio. 

Our idea is inspired by the formula for hyperbolic distance in $\H_r$. 
\[
d_H(q, x) = \arccosh\left(1 + \frac{\|q - x\|^2}{(1-\|q\|^2)(1-\|x\|^2)}\right) \eqref{eq:poincare_distance}
\]

If 2 points $x_1, x_2$ are such that $\|x_1\| \approx \|x_2\|$, then finding the nearer neighbor to $q$ reduces to minimizing $\|q - x\|$, which is a Euclidean nearest neighbor problem. Our overall scheme divides the dataset based on their Euclidean squared distance to the origin. Each batch of points in an annulus is organized into its own data structure that $\tO$ accesses. We probe relevant batches and return the best approximate nearest neighbor that we find from the different partitions.

In the preprocessing to divide $\D$, we take the multiplicative width of each annulus $w > 1$, and put into the $i$-th annulus, or partition, all data points $x$ such that $w^{i-1} \leq \frac{1}{1-\|x\|^2} \leq w^{i}$. The width $w$ controls the granularity of the $\ell_2$ norm at which we divide the dataset. 

The nearest neighbor algorithm, \bandedAlg probes different annuli using $\tO$ and returns the nearest hyperbolic neighbor from among $n_F$ returned by $\tO$ applied to each partition. One important detail is which partitions to probe and in which order. Algorithm \ref{alg:nn_banded_alg} offers one strategy. We first probe the band that the query falls into, $i$. Then we maintain two lists. The first list contains the indices higher than $i$ in sorted order, the other list contains the lowest. We choose from the top of the two lists, based on which choice maximizes the radius of the hyperbolic ball around $q$ that is completely covered by the union of bands probed so far as well as the new band under consideration. This is implemented in \decisionProbeBandMaxRadiusNS. We terminate based on \decisionProbeBandNS, which takes $n_H$, the best hyperbolic nearest neighbor found so far, and checks if there exists $x \in \H_r$ such that $w^{b-1} \leq \frac{1}{1-\|x\|^2}\leq w^b$ and also belongs to $\B_H(q, d_H(q, n_H))$.

\begin{algorithm}[!htp]
\caption{\bandedAlg}
\label{alg:nn_banded_alg}
\begin{algorithmic}[1]
\REQUIRE{query $q$, approximate Euclidean nearest neighbor oracle $\tO$, partitions $B$ with width $w$}
\STATE $i \leftarrow \lc\frac{-\log (1-\|q\|^2)}{\log w}\rc$ \text{ // the annulus that the query belongs to}
\STATE $ProbingListTop \leftarrow [i+1, i+2, \ldots]$ 
\STATE \text{ // the list of partitions arranged in probing order}
\STATE $ProbingListBottom \leftarrow [i-1, i-2, \ldots]$
\STATE \text{ // the list of partitions arranged in probing order}
\STATE $n_H \leftarrow \tO(q, B[i])$ \text{ // current best nearest neighbor candidate}
\STATE $dist_H  \leftarrow d_H(q, n_H)$
\STATE \text{ // hyperbolic distance of current best nearest neighbor candidate}
\WHILE {$\decisionProbeBand(q, n_H, w, ProbingListTop[0])$ or $\decisionProbeBand(q, n_H, w, ProbingListBottom[0])$}
\STATE ${ band = \decisionProbeBandMaxRadius(q, n_H, w, ProbingListTop[0], ProbingListBottom[0])}$
\STATE $n_F \leftarrow \tO(q, B[band])$
\IF {$d_H(q, n_F) < dist_H$}
\STATE $dist_H \leftarrow d_H(q, n_F)$
\STATE $n_H \leftarrow n_F$
\ENDIF
\ENDWHILE
\STATE \textbf{Return} $n_H$
\end{algorithmic}
\end{algorithm}

The routine calculations for $\decisionProbeBandMaxRadius$ and $\decisionProbeBand$ use elementary properties of hyperbolic geometry and are in the appendix. 

\bandedAlgNS, Algorithm \ref{alg:nn_banded_alg} provides the following approximation guarantee: 
\begin{theorem}
Using a $(1+\eps)$-Euclidean nearest neighbor oracle $\tO$ and a dataset split with a multiplicative width of $w$, \bandedAlg returns a hyperbolic approximate nearest neighbor $n_H$ to any query $q$ such that $d_H(q, n_H) \leq \sqrt{w}(1+\eps) d_H(q, n^*)$, where $n^*$ is the exact hyperbolic nearest neighbor. 
\end{theorem}

\begin{proof}
The true hyperbolic nearest neighbor, $n^*$ is organized into a bucket $j$ that \bandedAlg is guaranteed to probe. Suppose that instead of finding $n^*$, the algorithm finds $n_H$ in bucket $j$. The hyperbolic distance between the query $q$ and $n_H$, $D$, is upper bounded by 
\[
D \leq \arccosh \left(1 + \frac{2\|q - n_H\|^2 \cdot w^j}{(1-\|q\|^2)} \right)~,
\]
where the inequality comes from the guarantee that all elements $y$ in bucket $j$ satisfy $w^{j-1} \leq \frac{1}{1-\|y\|^2} \leq w^j$. This also implies that $\frac{\|q-n_H\|^2}{1-\|q\|^2} \geq \frac{\cosh(D) - 1}{2w^j}$. 

In the worst case, the true nearest neighbor $n^*$ is such that $\|q-n^*\|$ is much smaller than $\|q-n_H\|$ and also $\frac{1}{1-\|n^*\|^2}$ is much smaller than $\frac{1}{1-\|n_H\|^2}$. To make $\|q-n^*\|$ small, the worst case is that $n^*$ is actually the nearest neighbor in bucket $j$ to $q$. However, the guarantee of the approximate Euclidean oracle is that $\|q-n_H\| \leq (1+\eps) \|q-n^*\|$, so that $\|q-n^*\|^2 \geq \frac{\|q-n_H\|^2}{(1+\eps)^2}$. We also have that $d_H(q, n^*) = \arccosh\left(1 + \frac{2\|q - n^*\|^2}{(1-\|q\|^2)(1-\|n^*\|^2)}\right) \geq \arccosh\left(1 + \frac{\cosh(D)-1}{w(1+\eps)^2}\right)$. 

Now we want to analyze $\frac{D}{\arccosh\left(1 + \frac{\cosh(D)-1}{w(1+\eps)^2}\right)}$. \\
\begin{align*}
\arccosh\left(1 + \frac{\cosh(D)-1}{w(1+\eps)^2}\right) &= \arccosh \left(1 + \frac{\frac{e^D + e^{-D}}{2}-1}{w(1+\eps)^2} \right) \\
& = \arccosh \left(1 + \frac{\sum\limits_{i=1}^\infty \frac{D^{2i}}{(2i)!}}{w(1+\eps)^2} \right)  \geq \left( 1 + \sum\limits_{i=1}^\infty \frac{\left(\frac{D}{\sqrt{w}(1+\eps)} \right)^{2i}}{(2i)!}\right)\\
& = \arccosh(1 + \cosh(\frac{D}{\sqrt{w}(1+\eps)} -1)  = \frac{D}{\sqrt{w}(1+\eps)}
\end{align*}
Therefore, we conclude that $d_H(q, n_H) \leq \sqrt{w}(1+\eps)d_H(q, n^*)$. 
\end{proof}

The runtime of \bandedAlg depends on the number of partitions that are probed, which we now analyze. 

We first define $b_1, b_2, i_q$. Let $x = \argmax\limits_{z \in \B_H(q, d_H(q, n_H))} \|z\|$, and let $b_1 = \lc\frac{-\log(1-\|x\|^2)}{\log(w)}\rc$ denote the index of the partition that $x$ falls into, which is also the largest index that intersects this hyperbolic ball. Let $y = \argmin\limits_{z \in \B_H(q, d_H(q, n_H))} \|z\|$, and let $b_2 = \lf\frac{-\log(1-\|y\|^2)}{\log(w)}\rf$ be the index of the partition that $y$ falls into, which is also the smallest index possible that intersects the hyperbolic ball when 0 is not contained in this hyperbolic ball. When 0 is contained in the hyperbolic ball, the smallest partition index that intersects the hyperbolic ball is 1. 

\begin{lemma}
For a query $q$, suppose that $n_H$ is the approximate hyperbolic nearest neighbor returned by \bandedAlg. Further suppose that $d_H(0, q) > d_H(q, n_H)$. Then the number of partitions probed is $b_1 - b_2 + 1$. 
\end{lemma}

\begin{lemma}
For a query $q$, suppose that $n_H$ is the approximate hyperbolic nearest neighbor output of \bandedAlg. Further suppose that $d_H(0, q) \leq d_H(q, n_H)$. Then the number of partitions probed is $b_1$. 
\end{lemma}

\bandedAlg generalizes to return $K$ nearest neighbors with the worst case approximation guarantee for each neighbor if $\decisionProbeBandMaxRadius$ and $\decisionProbeBand$ use the distance of the $K$-th best nearest neighbor found so far.

One can design variants of \bandedAlg that differ in the probing sequence and probing criteria. We explore a randomized probing order in the appendix. This scheme uses a $(1+\epsilon, R)$-approximate Euclidean Near Neighbor Decision Oracle that gives a Yes/No answer for whether there is an element within distance $R$ to any point. We only probe a partition if the Decision Oracle says there is definitely a nearer neighbor in that partition than the current best. 

We show this variant has the same $\sqrt{w}(1+\eps)$-approximation guarantee as in \bandedAlgNS, and will fully search (using $\tilde{O}$) $\log(B)$ partitions in expectation, though the Decision Oracle could be applied to all partitions. 

Lastly, we show in the appendix that even with an exact Euclidean oracle $\O$, \bandedAlg is not guaranteed to return an exact hyperbolic nearest neighbor. 

\section{Evaluation}
We compare the techniques we develop to existing solutions. \cite{krauthgamer2006algorithms} presents an idea for hyperbolic nearest neighbor search but omits key implementation details (and we were unable to extract an efficient implementation from their proof). To our knowledge, the only other practical algorithms for this problem are nearest neighbor graph methods \cite{subramanya2019rand} \cite{malkov2018efficient} \cite{fu2019fast}, where the graph is constructed using hyperbolic distance. We compare the effectiveness of our technique against Vamana, a graph method that exhibited superior performance against the other in-class methods in the evaluation in \cite{subramanya2019rand}. As this family of algorithms does not come with any guarantee on the search quality, our experiments use a fixed sampling budget and compare the nearest neighbor found by the different algorithms under this budget. For the algorithms developed in this paper, if during the search the algorithm terminates before hitting this budget, we stop early. For the graph-based method, if the graph search terminates before hitting the budget, we initialize another round of search by starting at a different random initial point and search until we hit the budget.

We use a low-dimensional and a high dimensional dataset. Our queries are points that we withheld from the dataset. We solve the $K$-nearest neighbor ($K$-NN) problem for $K = 1,5$. We report the average recall for our batch of queries, defined as \# of the $K$ true nearest neighbors found / K. We also report the average approximation ratios and the max approximation ratio, where for $K > 1$, the ratio is computed pointwise: $d_H(q, n_k) / d_H(q, n^*_k)$. for each $k \in [K]$. For Vamana, we experimented with a range of hyperparameters and report the most favorable results. Our results largely show that our simple algorithms perform very well against Vamana. We find more exact nearest neighbors and we report better approximation ratios on average when we do not find the exact nearest neighbor. We also report the CPU running time of each our algorithms. 

\subsection{Low dimensional hyperbolic embeddings}
For the low-dimensional regime, we embed into 10 dimensions a dataset of 82,115 words from the WordNet noun hierarchy using the source code in \cite{nickel2017poincare}. As a sanity check, our trained embeddings achieve a rank of 4.739 and a MAP score of 0.811 in the reconstruction evaluation criteria as described in \cite{nickel2017poincare}, which is close to their reported results. 

Since we have exact oracles in this regime, we first consider whether it is efficient to use \recenteringAlg for real-world datasets by evaluating the number of calls to $\O$. We use a standard kd-tree\footnote{Source code for the kd-tree can be found at https://github.com/stefankoegl/kdtree.} as the underlying Euclidean oracle. To further optimize, we make the minor modification to the classic kd-tree -- whenever the algorithm solves for the Euclidean distance between a data point and $q$, we also solve for the hyperbolic distance. The traversal and termination criteria are all based on Euclidean distance; our modification also keeps track of the closest hyperbolic neighbor seen so far and returns that point. The analysis that we develop in this paper assuming a black box oracle still holds in this modified setting. 

In 2 independent trials, we withhold 800 queries from the dataset and record the number of calls to $\O$ that \recenteringAlg uses to find the exact nearest neighbor. We see that the number is low (Table \ref{table_kd_recenter}). Therefore, we use \recenteringAlg in our subsequent experiments to evaluate against Vamana. 

\begin{table}[h]
\caption{Statistics of number of calls to $\O$ in \recenteringAlg for sets of 800 queries}
\label{table_kd_recenter}
\centering
\begin{tabular}{ |p{0.8cm}||p{3cm}|p{0.8cm}|p{0.8cm}|p{0.8cm}| }
 \hline
Trial & Average \# of calls to $\O$ & SD &Min &Max \\
 \hline
 1   & 2.36    &0.51&   2 &4\\
 2 &   2.3  & 0.49   &2 & 4\\
 \hline
\end{tabular}
\end{table}

To compare against Vamana, we withhold 50 queries from the dataset. We report the results for the $1$-NN problem in Table~\ref{table:kd_tree_vs_graph}, and $5$-NN problem in Table ~\ref{table:kd_tree_vs_graph_5_NN}. We report a second trial with the same experimental setup in Tables \ref{table:kd_tree_vs_graph_2_trial} and \ref{table:kd_tree_vs_graph_5_NN_trial_2} We vary the budget of datapoints that the algorithm is able to search: 100, 500, 1000. After some hyperparameter tuning for Vamana, we use $L = 10$ and $R = 10$ and $\alpha = 1.5$, see \cite{subramanya2019rand} for more details on these hyperparameters. Our results show that \recenteringAlg with kd-tree generally finds more exact nearest neighbors than Vamana and approximate near neighbors with lower approximation ratios.  For $K > 1$, we use the KD tree to return $K$ nearest Euclidean neighbors, and we first recenter based on the nearest neighbor. When that termination criteria is hit, then we recenter based on the $K$-th nearest neighbor.

\begin{table}[h]
\caption{Trial 1. \recenteringAlg vs Vamana for $1$-NN search in the 10-dimensional noun hierarchy dataset} 
\centering
\scalebox{0.8}{
\begin{tabular}{|c|ccc|ccc|} 
\hline
& \multicolumn{3}{c|} {\recenteringAlg} & \multicolumn{3}{c|}{Vamana} \\
\hline
\# Samples & Recall & Avg Ratio & Avg Max & Recall & Avg Ratio & Avg Max\\
\hline 
100 &0.46&1.10 &1.66 &0.46 & 1.21  & 2.61 \\
 \hline 
500 & 0.7&1.036  & 1.37 &0.52 & 1.18 & 2.61 \\
 \hline 
1000 &0.84& 1.017& 1.37  &0.52 & 1.19 & 2.61 \\
\hline 
\end{tabular}}
\label{table:kd_tree_vs_graph}
\end{table}

\begin{table}[h]
\caption{Trial 2.  \recenteringAlg vs Vamana for $1$-NN search in the 10-dimensional noun hierarchy dataset}
\centering
\scalebox{0.8}{
\begin{tabular}{|c|ccc|ccc|}
\hline
& \multicolumn{3}{c|} {\recenteringAlg} & \multicolumn{3}{c|}{Vamana} \\
\hline
\# Samples & Recall & Avg Ratio & Avg Max & Recall & Avg Ratio & Avg Max\\
\hline 
100 &0.56&1.10 &1.69  &0.6 & 1.38 & 6.42  \\
 \hline 
500 & 0.78&1.035 & 1.59 &0.6 & 1.372 & 6.42 \\
 \hline 
1000 & 0.9 & 1.018& 1.27 &0.64& 1.29 & 6.40 \\
 
\hline 
\end{tabular}
}
\label{table:kd_tree_vs_graph_2_trial}
\end{table}

\begin{table}[h]
\caption{Trial 1. \recenteringAlg vs Vamana for $5$-NN search in the 10-dimensional noun hierarchy dataset} 
\centering 
\scalebox{0.8}{
\begin{tabular}{|c|ccc|ccc|} 
\hline 
& \multicolumn{3}{c|} {\recenteringAlg} & \multicolumn{3}{c|}{Vamana} \\
\hline
\# Samples & Recall & Avg Ratio & Avg Max & Recall & Avg Ratio & Avg Max\\
\hline 
100 & 0.23 & 1.15&1.24 & 0.420 &1.132  &1.24 \\
 \hline 
500 &0.48 &1.07 &1.12 &0.424 & 1.12 &1.23 \\
 \hline 
1000 & 0.59&1.04 &1.09 &0.452 &1.10  &1.20 \\
\hline 
\end{tabular}}
\label{table:kd_tree_vs_graph_5_NN}
\end{table}

\begin{table}[h]
\caption{Trial 2. \recenteringAlg vs Vamana for $5$-NN search in the 10-dimensional noun hierarchy dataset} 
\centering 
\scalebox{0.8}{
\begin{tabular}{|c|ccc|ccc|} 
\hline
& \multicolumn{3}{c|} {\recenteringAlg} & \multicolumn{3}{c|}{Vamana} \\
\hline
\# Samples & Recall & Avg Ratio & Avg Max & Recall & Avg Ratio & Avg Max\\
\hline 
100 &0.312 &1.173 &1.25 & 0.576  &1.21 &1.39\\
 \hline 
500 &0.584 & 1.07&1.12 &0.596  &1.20 &1.38\\
 \hline 
1000 &0.687 &1.044 &1.07 &0.576  &1.20 &1.39\\
\hline 
\end{tabular}}
\label{table:kd_tree_vs_graph_5_NN_trial_2}
\end{table}

\subsection{Approximate nearest neighbors for high dimensional hyperbolic embeddings}
For the high dimensional regime, we use provided embeddings from \cite{de2018representation} constructed using a higher dimensional extension of Sarkar's embedding algorithm \cite{sarkar2011low}. We use a dataset of 63,000 embeddings in 100 dimensions from the WordNet Hypernym noun hierarchy. Our Euclidean approximate nearest neighbor algorithm is the random hyperplane based scheme in \cite{datar2004locality}. We draw random hyperplanes uniformly from the unit sphere. For a random normal hyperplane $r$, the hash value of an element $x$ is $\frac{r \cdot x}{g}$, where $g$ is a granularity constant that determines how many equi-width segments we want to split the line segment $(-1, 1)$ into. As described in \cite{datar2004locality}, points that are close together tend to fall into the same segment. 

We use \bandedAlg with width $w = 3$, and 25 bands for extra tolerance. Each band $i$ containing normalized elements $x$ such that $3^{i-1} \leq \frac{1}{1-\|x\|^2} \leq 3^i$ is organized into an LSH data structure that uses 5 tables, with 15 random normalized hyperplanes per table, and with granularity $g=\min\{3^i, 10000\}$. We choose granularities based on data characteristics; in locality sensitive hashing, bucket widths are proportional to typical Euclidean nearest neighbor distances scaled by an appropriate function of the dimension and the number of random hyperplanes. The hyperplanes used for the LSH tables of each partition are the same. We probe buckets within distance 1 of the query bucket. 

Tables \ref{table:lsh_graph_experiments} and \ref{table:lsh_graph_experiments_5_NN} give the results for 49 queries withheld for the $1$-KNN and the $5$-KNN problems respectively. Tables \ref{table:lsh_graph_experiments_trial_2} and \ref{table:lsh_graph_experiments_5_NN_trial_2} give the results for a second trial of 38 queries withheld. 

After tuning for Vamana, we use $L = 40, R = 20, \alpha = 1.5$. Our results show that with \bandedAlg with LSH generally finds much more exact nearest neighbors than Vamana. Interestingly, for this type of very structured Sarkar embeddings, we outperform Vamana by a larger margin than for the trained embeddings in the previous experiment. 

\begin{table}[h]
\centering
\caption{Trial 1. \bandedAlg vs Vamana for $1$-NN in the 100-dimensional noun hierarchy dataset} 
\scalebox{0.8}{
\centering 
\begin{tabular}{|c|ccc|ccc|} 
\hline 
& \multicolumn{3}{c|} {\bandedAlg} & \multicolumn{3}{c|}{Vamana} \\
\hline
\# Samples & Recall & Avg Ratio & Avg Max & Recall & Avg Ratio& Avg Max \\
\hline 
100 & 0.43 & 2.01 &8.99 &0.04 & 3.66  & 8.68  \\
 \hline 
500 &0.71 &1.39 & 5.91 &0.18 & 2.001  & 4.49 \\
 \hline 
1000 & 0.90& 1.053&1.81 &0.39 & 1.52 & 3.50 \\
\hline
\end{tabular}}
\label{table:lsh_graph_experiments}
\end{table}

\begin{table}[h]
\caption{Trial 2. \bandedAlg vs Vamana for $1$-NN in the 100-dimensional noun hierarchy dataset}
\centering 
\scalebox{0.8}{
\begin{tabular}{|c|ccc|ccc|}
\hline 
& \multicolumn{3}{c|} {\bandedAlg} & \multicolumn{3}{c|}{Vamana} \\
\hline
\# Samples & Recall & Avg Ratio & Avg Max & Recall & Avg Ratio& Avg Max \\
\hline 
100 &0.58 &1.63 & 6.05 &0.05 & 3.413 & 8.53  \\
 \hline 
500 &0.71 &1.28 & 4.74 &0.16 & 1.798& 5.06 \\
 \hline 
1000 &0.92 &1.08 &3.68&0.32 & 1.55 & 3.98 \\
\hline
\end{tabular}
}
\label{table:lsh_graph_experiments_trial_2}
\end{table}

\begin{table}[h]
\centering
\caption{Trial 1. \bandedAlg vs Vamana for $5$-NN in the 100-dimensional noun hierarchy dataset}
\scalebox{0.8}{
\centering
\begin{tabular}{|c|ccc|ccc|}
\hline 
& \multicolumn{3}{c|} {\bandedAlg} & \multicolumn{3}{c|}{Vamana} \\
\hline
\# Samples & Recall & Avg Ratio & Avg Max & Recall & Avg Ratio& Avg Max \\
\hline
100 & 0.32 &1.57 &2.19 &0.016 &2.76  & 3.76  \\
 \hline 
500 &0.65 & 1.20& 1.41 & 0.09&1.73  & 2.35  \\
 \hline 
1000 &0.81 &1.052 & 1.12&0.187 & 1.40 & 1.85  \\
\hline 
\end{tabular}}
\label{table:lsh_graph_experiments_5_NN}
\end{table}

\begin{table}[h]
\centering
\caption{Trial 2. \bandedAlg vs Vamana for $5$-NN in the 100-dimensional noun hierarchy dataset}
\scalebox{0.8}{
\centering
\begin{tabular}{|c|ccc|ccc|}
\hline
& \multicolumn{3}{c|} {\bandedAlg} & \multicolumn{3}{c|}{Vamana} \\
\hline
\# Samples & Recall & Avg Ratio & Avg Max & Recall & Avg Ratio& Avg Max \\
\hline
100 & 0.410&1.40& 1.75&0.063 & 2.79 &3.66 \\
 \hline 
500 &0.668 &1.22&1.36 &0.147 & 1.65 & 2.13 \\
 \hline
1000 &0.784 &1.14&1.244 &0.236 & 1.521 & 1.97 \\
\hline
\end{tabular}}
\label{table:lsh_graph_experiments_5_NN_trial_2}
\end{table}

\subsection{Running Time}
We report the running time ratio for the $5$-KNN problem in Table \ref{table:running_time_5nn}, where the ratio is defined as the time for our techniques / Vamana's running time (so lower is better). Overall, our methods are faster than Vamana. This difference is likely because our algorithms have termination criteria that may not exhaust the given sampling budget, and so we stop early, whereas for the graph based Vamana, we maximize the budget. In the latter case, we do so because Vamana (and other in-class graph algorithms \cite{malkov2018efficient} \cite{fu2019fast}) perform better when the graph is searched multiple times using different initial points (even so, there are no theoretical guarantees).
\begin{table}[h]
\caption{Running time Ratios}
\centering
\scalebox{1}{
\begin{tabular}{|c|ccc|ccc|}
\hline
& \multicolumn{3}{|c|} {Low dimensional} & \multicolumn{3}{c|}{High dimensional} \\
 \hline
\# samples&100& 500 &1000 & 100& 500 &1000\\
\hline
Ratio &0.07 & 0.03 & 0.02 & 0.017 & 0.006 &0.0018 \\
\hline
\end{tabular}}
\label{table:running_time_5nn}
\end{table}

\section{Conclusion}
We consider the problem of nearest neighbor search for hyperbolic embeddings. We give theoretical guarantees and hardness results for our techniques. Experimental validation shows the effectiveness of our techniques against baseline methods. 

\appendix
\section{Details for \recenterBallsNS}
We now provide the helper routine to recenter the hyperbolic ball to its Euclidean center, \recenterBallsNS. The reasoning that \recenterBalls will return the Euclidean center of the hyperbolic ball is as follows: 
\begin{itemize}
\item Hyperbolic distance is additive on the line. 
\item $t_1c_H, t_2c_H$ and $c_H$ are collinear. Moreover, $d_H(t_1c_H, c_H) = d_H(t_2c_H, c_H) = r$ and therefore, $d_H(t_1c_H, t_2c_H) = 2r$ and so $t_1c_H$ and $t_2c_H$ are points on the sphere whose distance achieves the largest possible according to the hyperbolic metric, and so they form the endpoints of a line segment that passes through the center of the Euclidean circle. So we can take their average to find the center. 
\end{itemize}

\begin{algorithm}[h]
\caption{\recenterBalls}
\begin{algorithmic}[1]
\REQUIRE{hyperbolic center $c_H$, radius of hyperbolic ball $r$}
\IF {$c_H = \vec{0}$}
\STATE \textbf{Return} $\vec{0}$
\ENDIF
\STATE Find scalar $t_1$ such that $t_1 \|c_H\|_2 = \tanh \left( \frac{d_H(0, c_H) + r}{2} \right)$
\STATE Find scalar $t_2$ such that $t_2 \|c_H\|_2 = \tanh \left( \frac{d_H(0, c_H) - r}{2} \right)$
\STATE \textbf{Return} $\frac{t_1c_H + t_2c_H}{2}$
\end{algorithmic}
\end{algorithm}

\subsubsection{Best Case Configuration for \recenteringAlg}
A best case configuration is the following. The query, $q$, is point $(0, 0.99)$. Suppose now that the true hyperbolic nearest neighbor, $n_H$, is at point $(0,.981)$ and there is another point $n_E$, at $(0, 0.998)$, which is the Euclidean nearest neighbor to $q$. 

At the first iteration, the Euclidean nearest neighbor oracle $\O$ returns $n_E$. Then the hyperbolic circle radius is: 
\[d_H(q, n_E) = \arccosh \left(1 + \frac{2 \|q-n_E\|^2}{(1-\|q\|^2)(1-\|n_E\|^2)} \right)\]

The other boundary of the hyperbolic ball in the direction of the query $q$, denoted $n_B$ is a point of the form $(0, b)$. We solve for $b$ by noticing that $n_B$ satisfies: 
\begin{align*} d_H(q, n_B) &= \arccosh \left(1 + \frac{2 \|q-n_B\|^2}{(1-\|q\|^2)(1-\|n_B\|^2)} \right) \\
& = \arccosh \left(1 + \frac{2 (0.99 - b)^2}{(1-(0.99)^2)(1-b^2)} \right)
\end{align*}
Equating the expression to $d_H(q, n_E)$ gives us that $b \approx .912252$. 
\\\\
Therefore, the Euclidean center of this hyperbolic circle, denoted $q_{new}$, is $(0, 0.9551260)$. 
\\\\
Now suppose additionally we have $k-2$ points on the $y$-axis between $.912252$ and $.928$, for arbitrary $k$. Clearly then $n_E$ is the $k$-th hyperbolic nearest neighbor of $q$ but \recenteringAlg will return the hyperbolic nearest neighbor in  3 rounds. 

\section{Integration with approximate Euclidean nearest neighbor oracles}
\begin{lemma}
For any $\eps > 0$, \binarySearchAlg using a $(1+\eps)$-approximate Euclidean nearest neighbor oracle $\tO$ can return an approximate hyperbolic nearest neighbor with an arbitrarily bad approximation ratio. 
\end{lemma}

\begin{proof}
 As before, let 
$$n_E = (0, 1-\gamma)$$
$$q = \left(0, 1- \left(\frac{\gamma + \delta}{2}\right)\right)$$
$$n_H = (0, 1-\delta)$$
Suppose that $\frac{d_H(q, n_E)}{d_H(q, n_H)} = S$ for some very high $S$. Then we want to show that if $\delta$ is sufficiently high, \binarySearchAlg will return $n_E$ and fail to find $n_H$, leading to a bad approximation ratio of $S$.

Clearly, $RL = S(d_H(q, n_H))^2$ in this case, so $\sqrt{RL} = \sqrt{S}d_H(q, n_H)$. We want to find $T_1 = (0, t_1)$ and $T_2 = (0, t_2)$ such that $d_H(q, T_1) = d_H(q, T_2) = \sqrt{RL}$. \binarySearchAlg will call the $(1+\eps)$-Euclidean oracle to find the nearest neighbor of $n_c = \frac{T_1 + T_2}{2}$. 

For clarity, let's say that $q = (0, y), n_E = (0, y+r), n_H = (0, y-r)$, where $y>0, r>0$. 
For \binarySearchAlg to fail, the condition we want is: 
\[
y - r > \frac{t_1 + t_2}{2} + \frac{d_E(n_E, n_c)}{1+\eps} = \frac{t_1 + t_2}{2} + \frac{y + r - \frac{t_1 + t_2}{2}}{1+\eps}
\]

This condition is equivalent to: 
\[
\frac{t_1 + t_2}{2} < \frac{\eps(y-r)-2r}{\eps}
\]

Let $D = \sqrt{RL}$. One can calculate that 
\[
t_1 = \frac{\sinh\left(\frac{D}{2}\right) - y \cosh \left( \frac{D}{2}\right)}{y\sinh\left(\frac{D}{2}\right) - \cosh \left( \frac{D}{2}\right)}
\]
\[
t_2 = \frac{\sinh\left(\frac{D}{2}\right) + y \cosh \left( \frac{D}{2}\right)}{y\sinh\left(\frac{D}{2}\right) + \cosh \left( \frac{D}{2}\right)}
\]

Therefore, we have: 
\begin{align*}
\frac{t_1 + t_2}{2} &= \frac{y \left( \sinh^2 \left( \frac{D}{2}\right) - \cosh^2 \left( \frac{D}{2} \right) \right)}{y^2 \sinh^2 \left( \frac{D}{2}\right) - \cosh^2 \left( \frac{D}{2}\right)} \\
&= \frac{y}{\cosh^2 \left( \frac{D}{2}\right) - y^2 \sinh^2 \left( \frac{D}{2}\right)} \\
& = \frac{y}{\cosh^2 \left( \frac{D}{2}\right) - y^2 \left( \cosh^2 \left( \frac{D}{2}\right) - 1\right)} \\
& = \frac{y}{(1-y^2)\cosh^2 \left( \frac{D}{2}\right) + y^2 } \\
& = \frac{2y}{(1-y^2)\left( 1+ \cosh(D)\right) + 2y^2 } \\
& = \frac{2y}{1-y^2 + \cosh(D)(1-y^2) + 2y^2} \\
& = \frac{2y}{1+y^2 + \cosh(D)(1-y^2)} \\
& \leq \frac{2y}{1+y^2 + (1-y^2)\frac{e^D}{2}} \\
& \leq \frac{2y}{(1-y^2)(1 + \frac{e^D}{2})} \\
& \leq \frac{4y}{(1-y^2)(e^D)}
\end{align*}

Note that $D = \sqrt{RL} = \sqrt{S} d_H(q, n_H)$. 

Remember that we have: 
\begin{align*}
d_H(n_H, q) &= \arccosh \left(1+ \frac{2 \left(\frac{\delta - \gamma}{2}\right)^2}{\left(2\delta - \delta^2\right)\left(\gamma + \delta - \left(\frac{\gamma + \delta}{2}\right)^2\right)} \right) \\
& \geq \arccosh \left(1+ \frac{2 \left(\frac{\delta - \gamma}{2}\right)^2}{\left(2\delta\right)\left(\gamma + \delta\right)} \right) \\
& \geq \arccosh \left( 1 + \frac{1}{\delta} \cdot \frac{\delta}{8}\right) \\
& \geq \arccosh \left( 1 + \frac{1}{8} \right) \geq 0.49
\end{align*}
where we again use that $\gamma$ is sufficiently small that $\frac{\left( \delta - 2\gamma\right)}{\delta +\gamma} \geq \frac{1}{2}$. 

This implies that $D \geq 0.49\sqrt{S}$, so $e^D \geq e^{ 0.49\sqrt{S}}$. 

So we want:
\[
\frac{4y}{(1-y^2)(e^D)} \leq \frac{\eps(y-r)-2r}{\eps} = y-r - \frac{2}{\eps}r
\]
This is equivalent to: 
\[
r\left(1+ \frac{2}{\eps}\right) \leq y\left(1-\frac{4}{(1-y^2)(e^D)}\right)
\]

Remember that $r = \frac{\delta - \gamma}{2} < \frac{\delta}{2}$, so we have: 
\[
r\left(1+ \frac{2}{\eps}\right) \leq \frac{3\delta}{2\eps}
\]

Now to focus on the right hand side, if we have: 
\[
\frac{4}{e^D} < \frac{1}{2} (1-y^2) ~, 
\]

then we have
\[
y\left(1-\frac{4}{(1-y^2)(e^D)}\right) \geq \frac{y}{2}
\]

Also we can say that $y =  1- \left(\frac{\gamma + \delta}{2}\right) > \frac{1}{2}$, so that $y\left(1-\frac{4}{(1-y^2)(e^D)}\right) > \frac{1}{4}$. 

Then for a given $(1+\eps)$-approximate Euclidean oracle, as long as $\delta$ is small enough that $\frac{3\delta}{2\eps} \leq \frac{1}{4}$ or $\delta < \frac{\eps}{6}$, then \binarySearchAlg will fail. 

Now to see how to satisfy the constraint that $\frac{4}{e^D} < \frac{1}{2} (1-y^2)$. 

Note that 
\begin{align*}
\frac{1}{2} (1-y^2) &= \frac{1}{2} \left( \delta + \gamma - \left( \frac{\gamma + \delta}{2}\right)^2\right) \\
&\geq \frac{1}{2}\left( \frac{\gamma + \delta}{2}\right)^2 \\
& \geq \frac{\delta^2}{8}
\end{align*}

From before, we had that $\frac{4}{e^D} \leq \frac{4}{e^{ 0.49\sqrt{S}}}$. 

Then a sufficient condition is that $S$ is large enough that $\frac{4}{e^{ 0.49\sqrt{S}}} \leq \frac{\delta^2}{8}$. 
\end{proof}

\section{\bandedAlg}

\subsection{Details for \bandingDatasetAlgNS}
We first describe the partitioning algorithm to divide the dataset into bands based on Euclidean norm. \bandingDatasetAlgNS, Algorithm \ref{alg:nn_make_bands} is the formal pre-processing algorithm to divide the dataset. The algorithm works by taking in the largest possible norm that one wishes to support; for a given dataset, this could be the norm of the largest data point or a norm slightly higher than that for extra tolerance, as well as the multiplicative width of each annulus $w$, for $(w > 1)$. The width $w$ controls the granularity at which we divide the dataset based on $1 - \|x\|^2$. The $i$-th annulus, or partition, contains all data points $x$ such that $w^{i-1} \leq \frac{1}{1-\|x\|^2} \leq w^{i}$. 
\begin{algorithm}[ht]
\caption{\bandingDatasetAlg}
\label{alg:nn_make_bands}
\begin{algorithmic}[1]
\REQUIRE{dataset $\D$, multiplicative width of annulus, $w$, largest possible norm to support, $L$}
\STATE \text{num\_bands} $\leftarrow \lc\frac{-\log (1-\|L\|^2)}{\log w}\rc$
\STATE \text{Initialize (num\_bands -1) partitions to organize datasets into, denote $B[i]$ as the $i$-th partition}. 
\FORALL{$x \in \D $}
\STATE $i = \lc\frac{-\log (1-\|x\|^2)}{\log w}\rc$
\STATE \text{Insert $x$ into $B[i]$}
\ENDFOR
\STATE \textbf{Return} $B$
\end{algorithmic}
\end{algorithm}

\subsection{Details for \decisionProbeBand}

We now describe the helper routine for \bandedAlg that determines whether to probe a band (Algorithm \ref{alg:probe_bucket}). The idea behind \decisionProbeBand is very simple. It takes in the center, and a point on the intended hyperbolic ball, which in our case is the query $q$ and the current best nearest neighbor, $n_H$, respectively, as well as the multiplicative width of the buckets and the bucket index to evaluate. The point $x$ with the largest possible Euclidean norm of any of the points in this ball satisfies $d_H(x, 0) = d_H(q, n_H) + d_H(0, q)$. Moreover, if $x$ were of the form $t_1c_H$ for some scalar $t_1$, since hyperbolic distance is additive on the line, we also satisfy that $t_1c_H$ is on the boundary of the ball. Therefore, we just have to solve for this $t_1$ and calculate the bucket index $j$ that $t_1c_H$ would ordinarily partition to. If the bucket index under consideration $b$ is greater than $i$ (the bucket index that the query partitions to), we should probe $b$ if $b < j$. If $b < i$, then we do the same calculation but for the reverse situation where we analyze the smallest possible Euclidean norm of any point in the hyperbolic ball. One small difference is that the origin might be contained in this ball, in which case the $t_2$ might be negative. In that case, we should search all buckets with indices smaller than $i$. 

\begin{algorithm}[ht]
\caption{\decisionProbeBand}
\label{alg:probe_bucket}
\begin{algorithmic}[1]
\REQUIRE{hyperbolic center $c_H$, point on the boundary of hyperbolic ball $p$, multiplicative width of annulus $w$, bucket index to evaluate $b$}
\STATE $i \leftarrow \lc\frac{-\log (1-\|c_H\|^2)}{\log w}\rc$
\IF {$p = NULL$}
\STATE \textbf{Return} \text{True}
\ELSIF {$b \geq i$}
\STATE Find scalar $t_1$ such that $t_1 \|c_H\|_2 = \tanh \left( \frac{d_H(0, c_H) + d_H(c_H, p)}{2} \right)$
\STATE $j \leftarrow \lc\frac{-\log(1-\|t_1c_H\|^2)}{\log(w)}\rc$
\IF {$b \leq j$}
\STATE \textbf{Return} \text{True}
\ENDIF
\ELSE
\STATE Find scalar $t_2$ such that $t_2 \|c_H\|_2 = \tanh \left( \frac{d_H(0, c_H) - d_H(c_H, p)}{2} \right)$
\IF {$t_2 \leq 0$}
\STATE \textbf{Return} \text{True}
\ELSE
\STATE $j \leftarrow \lf\frac{-\log(1-\|t_2c_H\|^2)}{\log(w)}\rf$
\IF {$b \geq j$}
\STATE \textbf{Return} \text{True}
\ENDIF
\ENDIF
\ENDIF
\STATE \textbf{Return} \text{False}
\end{algorithmic}
\end{algorithm}

\subsection{Details for \decisionProbeBandMaxRadius}
We describe the helper routine that decides whether the algorithm should search in band $b_1$ or $b_2$, when the algorithm is guaranteed to have already searched in bands $i, i+1 \ldots b_1 - 1$, and $i-1, i-2 \ldots b_2+1$, where $i$ is the band index that the query falls into. The overall idea is that when deciding which next band to probe, we choose the band which maximizes the radius of the hyperbolic ball around $q$ that is completely covered by the union of bands probed so far as well as the new band under consideration.

\begin{algorithm}[ht]
\caption{\decisionProbeBandMaxRadius}
\label{alg:choose_bucket}
\begin{algorithmic}[1]
\REQUIRE{hyperbolic center $c_H$, current best neighbor $n_H$, multiplicative width of annulus $w$, bucket index to evaluate $b_1$, $b_2$, wlog $b_1 > b_2$}
\STATE $d_1 \leftarrow -\infty$
\STATE $d_2 \leftarrow -\infty$
\IF {$\decisionProbeBand(c_H, n_H, w, b_1)$}
\STATE Find scalar $t_1$ such that $\frac{1}{1-\|t_1 c_H\|^2} = w^{b_1}$
\STATE Find scalar $t_2$ such that $\frac{1}{1-\|t_2 c_H\|^2} = w^{b_2}$
\STATE $d_1 \leftarrow \min \{ d_H(c_H, t_1c_H), d_H(c_H, t_2 c_H)\}$
\ENDIF

\IF {$\decisionProbeBand(c_H, n_H, w, b_2)$}
\STATE Find scalar $t_3$ such that $\frac{1}{1-\|t_3 c_H\|^2} = w^{b_1-1}$
\STATE Find scalar $t_4$ such that $\frac{1}{1-\|t_4 c_H\|^2} = w^{b_2-1}$
\STATE $d_2 \leftarrow \min \{ d_H(c_H, t_3c_H), d_H(c_H, t_4 c_H)\}$
\ENDIF
\IF {$d_1 \geq b_2$}
\STATE \textbf{Return} \text{$b_1$}
\ENDIF
\STATE \textbf{Return} \text{$b_2$}
\end{algorithmic}
\end{algorithm}

\subsection{\bandedRandAlg}
The probing strategy for \bandedAlg in the worst case (for large hyperbolic distances between $q$ and $n_H$) would probe many buckets, possibly all the buckets. To reduce the number of buckets probed, we introduce a randomized algorithm that orders the buckets uniformly at random among all possible permutations, and calls the Euclidean nearest neighbor oracle $\tO$ on the first bucket on the list to find a starting nearest neighbor candidate with hyperbolic radius $r$ to the query. On subsequent buckets, we first use a decision oracle to determine whether that bucket will definitely contain an element closer to $q$ than the current best. If the decision oracle says yes, then we do a full probe on that bucket. Otherwise we move onto the next bucket on the list. The advantage here is that a query to the decision oracle can be very fast, so if \bandedAlg would do a full probe on all the buckets, this randomized algorithm would in expectation do a full probe on a small number of buckets. However, this algorithm uses a decision oracle, which is not always available, or efficient. We first define the decision oracle. 

\begin{definition}[$(1+\eps, R)$-approximate Euclidean Near Neighbor Decision Oracle, $\tDO$]
The $(1+\eps, R)$-approximate Euclidean Near Neighbor Oracle, $\tDO$ takes as input a query $q$, radius of interest $R$, approximation factor $\eps > 0$, and a dataset of elements $\D$. If the Euclidean nearest neighbor to $q$, denoted $n_E$, satisfies $\|q-n_E\| \leq R$, this oracle returns a certificate element $x'$ such that $\|x' - q\| \leq (1+\eps) R$. 
\end{definition}

It is actually possible to build a $(1+\eps)$-approximate Euclidean nearest neighbor oracle by calling on the $(1+\eps, R)$-approximate Euclidean Near Neighbor Decision Oracle multiple times using successively smaller values of $R$ in a binary search fashion. The query times for the decision oracle are typically smaller than for the approximate near neighbor oracles (since we are not searching for the nearest, just for something nearer than $R$), the saving is about a factor logarithmic in $n$. 

\begin{algorithm}[!htp]
\caption{\bandedRandAlg}
\label{alg:nn_banded_rand_alg}
\begin{algorithmic}[1]
\REQUIRE{query $q$, $(1+\eps)$-approximate Euclidean NN oracle $\tO$, $(1+\eps, R)$-approximate decision oracle $\tDO$, buckets $B$ with width $w$}
\STATE $ProbingList \leftarrow Unif(B)$ \text{ // the list of buckets arranged in a random order}
\STATE $n_H \leftarrow \text{NULL}$ \text{ // current best nearest neighbor candidate}
\STATE $dist_H = \infty$ \text{ // hyperbolic distance of current best nearest neighbor candidate}
\FOR{buckets $b$ in ProbingList}
\STATE $R \leftarrow \sqrt{\frac{\left( \frac{\cosh(dist_H)-1}{2}\right) (1-\|q\|^2)}{w^i (1+\eps)^2}}$
\IF {$dist_H = \infty$ or $\tDO(q, R, B[b]) = YES$}
\STATE $n_F \leftarrow \tO(q, B[b])$
\IF {$d_H(q, n_F) < dist_H$}
\STATE $dist_H \leftarrow d_H(q, n_F)$
\STATE $n_H \leftarrow n_F$
\ENDIF
\ENDIF
\ENDFOR
\STATE \textbf{Return} $n_H$
\end{algorithmic}
\end{algorithm}

We first analyze the approximation guarantee of this \bandedRandAlg. Then we give the analysis for the expected number of full probes made by the approximate nearest neighbor oracle $\tO$. 
\begin{theorem}
Using a $(1+\eps)$-Euclidean nearest neighbor oracle $\tO$, a $(1+\eps, R)$-Euclidean near neighbor decision oracle and a dataset split with a multiplicative width of $w$, \bandedRandAlg returns a hyperbolic approximate nearest neighbor $n_H$ to any query $q$ such that $d_H(q, n_H) \leq \sqrt{w}(1+\eps) d_H(q, n^*)$. 
\end{theorem}

\begin{proof}

Suppose that the current best nearest neighbor candidate, $n_H$ has hyperbolic distance $D$ to the query. Further suppose we are looking at the $i$-th bucket. This bucket contains elements $y$ such that $\frac{1}{w^i} \leq 1-\|y\|^2 \leq \frac{1}{w^{i-1}}$. We want to ask this bucket if it contains an element $x$ such that $d_H(q, x) < D$. 

So we want:  
\[
d_H(q, x) = \arccosh \left(1 + \frac{2\|q-x\|^2}{(1-\|q\|^2)(1-\|x\|^2} \right) \leq D
\]
This implies that
\begin{align*}
\|q-x\|^2 &\leq \left( \frac{\cosh(D)-1}{2}\right) (1-\|q\|^2) (1-\|x\|^2) \\
&\leq \frac{\left( \frac{\cosh(D)-1}{2}\right) (1-\|q\|^2)}{w^i}
\end{align*}

So if bucket $i$ contains an element $x$ such that
\[
\|q-x\| \leq \sqrt{\frac{\left( \frac{\cosh(D)-1}{2}\right) (1-\|q\|^2)}{w^i}}
\]
then $x$ is definitely a nearer neighbor to $q$ than $n_H$. 

But since we are using a $(1+\eps)$-approximate nearest neighbor oracle, to guarantee that the oracle only returns an element if bucket $i$ is guaranteed to contain a nearer neighbor, we let $R = \sqrt{\frac{\left( \frac{\cosh(D)-1}{2}\right) (1-\|q\|^2)}{w^i (1+\eps)^2}}$.

Now to analyze the approximation factor. Some error could be introduced in the fact that the decision oracle could have said ``NO" but the bucket actually contained a closer element, but this closer element was just slightly closer to $q$ than $n_H$. Say that this happened and we just missed $n^*$. Then clearly, 
\[
\|q-n^*\| \geq \sqrt{\frac{\left( \frac{\cosh(D)-1}{2}\right) (1-\|q\|^2)}{w^i(1+\eps)^2}}
\]
Moreover, $\frac{1}{1-\|n^*\|} \geq w^{i-1}$

Therefore, 
\begin{align*}
d_H(q, n^*) &= \arccosh \left(1 + \frac{2\|q-n^*\|^2}{(1-\|q\|^2)(1-\|n^*\|^2} \right) \\
& \geq \arccosh \left(1 + \frac{2 w^{i-1} \|q-n^*\|^2}{(1-\|q\|^2)} \right) \\
& \geq \arccosh \left( 1 + \frac{\cosh(D) -1}{w(1+\eps)^2} \right)
\end{align*}

The rest of the proof follows similarly to the proof for \bandedAlgNS. 
\end{proof}

Now we want to provide an analysis on the expected number of invocations of the approximate nearest neighbor oracle $\tO$. We have the following theorem: 
\begin{lemma}
Suppose that there are $N$ buckets in total, and the probing order is selected uniformly at random among all the possible permutations of the $N$ buckets. Then the expected number of invocations of the approximate nearest neighbor oracle $\tO$ is $O(\ln N)$.
\end{lemma}
\begin{proof}
We proceed with a proof by induction. The base case when $N=1$ holds. Now suppose that for $k= 2, 3, \ldots N-1$ buckets, the expected number of invocations is $\sum\limits_{n=1}^{k} \frac{1}{n}$.
Now let us consider the case when we have $N$ buckets. First of all, we always probe the first bucket. Now suppose the hyperbolic nearest neighbor to $q$ in the first bucket is the $k$-th hyperbolic nearest neighbor to $q$ among the entire dataset. Then we subsequently have to probe at most $k-1$ buckets, so the problem has been reduced to the subproblem of solving for the number of expected probes where the total number of buckets is $k-1$, which by our inductive assumption is $\sum\limits_{j=1}^{k-1} \frac{1}{j}$. This event happens with probability $\frac{1}{N}$. Now, summing over all possible values of $k$ gives us the following expression: 
\[
1 + \frac{1}{N} \sum\limits^{N-1}_{k=1} \sum\limits^k_{j=1} \frac{1}{j} ~.
\]

Also note that by this reasoning combined with the inductive hypothesis gives that 
\[
1 + \frac{1}{N-1}\sum\limits^{N-2}_{k=1} \sum\limits^k_{j=1} \frac{1}{j} = \sum\limits_{n=1}^{N-1} \frac{1}{n} ~.
\]
Now to evaluate: 
\begin{align*}
1 + \frac{1}{N} \sum\limits^{N-1}_{k=1} \sum\limits^k_{j=1} \frac{1}{j} &= 1 + \frac{N-1}{N} \cdot \frac{1}{N-1}\sum\limits^{N-2}_{k=1} \sum\limits^k_{j=1} \frac{1}{j} + \frac{1}{N} \sum\limits^{N-1}_{j=1} \frac{1}{j} \\
&=1 + \frac{N-1}{N} \cdot \sum\limits_{n=2}^{N-1} \frac{1}{n} + \frac{1}{N} \sum\limits^{N-1}_{j=1} \frac{1}{j} \\
&=1 + \sum\limits_{n=2}^{N-1} \frac{1}{n} + \frac{1}{N} \\
&=\sum\limits_{n=1}^N \frac{1}{n}
\end{align*}
\end{proof}

Then we come to the final runtime guarantee of \bandedRandAlgNS. 
\begin{theorem}[Runtime of \bandedRandAlgNS]
The expected runtime of \bandedRandAlg is $O(\mathcal{T} \cdot  \ln N +  \mathcal{T}_\mathcal{D} \cdot N)$, 
where $\mathcal{T}$ is the runtime for one invocation of $\tO$ and $\mathcal{T}_\mathcal{D}$ is the runtime for one invocation of the decision oracle $\tDO$ and $N$ is the total number of buckets. 
\end{theorem}
\subsection{\bandedAlg cannot return an exact hyperbolic nearest neighbor with an exact Euclidean oracle $\O$}
We provide a simple example demonstrating that even with an exact Euclidean oracle $\O$, \bandedAlg is not guaranteed to return an exact hyperbolic nearest neighbor. However, if the dataset has already been divided into buckets according to \bandingDatasetAlgNS, one can additionally leverage the recentering idea that forms the core of \recenteringAlg to return an exact hyperbolic nearest neighbor. We leave the implementation details to the reader. 

The example is as follows. Suppose we have a dataset of two points, $n^* = (0, 0.5)$ and $n_E = (0.15, 0.55)$ and the query $q$ is $(0,0.99)$. Straightforward calculation shows that:
\[
\frac{1}{1-\|n^*\|^2} \approx 1.33
\]
and 
\[
\frac{1}{1-\|n_E\|^2} \approx 1.48 ~. 
\]
Therefore for $w \geq 1.5$, \bandingDatasetAlg will designate them into the same bucket. 

We also remark that the hyperbolic nearest neighbor is $n^*$, since $d_H(q, n^*) \approx 4.19$ and $d_H(q, n_E) \approx 4.384$. 

However, the Euclidean nearest neighbor of $q$ is $n_E$, not $n^*$, with $\|q - n_E\| \approx .464 $ and $\|q - n^*\| = .49$. Therefore, \bandedAlg will return $n_E$ when using an exact Euclidean nearest neighbor oracle, which is an approximate nearest neighbor to the query.  

This example relies crucially on the fact that depending on the placement of the hyperbolic nearest neighbor $n^*$ on the Poincare disk, the hyperbolic ball around the query $q$ with radius $d_H(q, n^*)$, call it $B_H(q, d_H(q, n^*))$, can be completely contained in $B_E(q, d_E(q, n^*))$, the Euclidean ball around $q$ with radius $d_E(q, n^*)$. When this is true, for any predetermined value of $c$, one can find a set of $q, n^*, n_E$ where \bandedAlg cannot guarantee an exact nearest neighbor even when using an exact Euclidean nearest neighbor oracle. 
\bibliographystyle{alpha}
\bibliography{hyperbolic}

\end{document}